\documentclass[aps,pra,onecolumn,notitlepage,superscriptaddress,showpacs,nofootinbib]{revtex4-1}
\usepackage{amsmath}
\usepackage{amssymb}
\usepackage{amsfonts}
\usepackage{enumerate}
\usepackage{mathrsfs}
\usepackage{graphicx}
\usepackage{epstopdf}
\usepackage{enumerate}
\usepackage{changepage}
\usepackage{bm}
\usepackage{multirow}
\usepackage{lipsum}
\usepackage{booktabs}
\usepackage{graphicx}
\usepackage{float}

\newtheorem{theorem}{Theorem}
\newtheorem{definition}{Definition}

\newtheorem{lemma}{Lemma}
\newtheorem{proposition}{Proposition}
\newtheorem{conjecture}{Conjecture}
\newtheorem{example}{Example}
\newtheorem{corollary}{Corollary}

\def\bcj{\begin{conjecture}}
	\def\ecj{\end{conjecture}}
\def\bcr{\begin{corollary}}
	\def\ecr{\end{corollary}}
\def\bd{\begin{definition}}
	\def\ed{\end{definition}}
\def\bea{\begin{eqnarray}}
	\def\eea{\end{eqnarray}}
\def\bem{\begin{enumerate}}
	\def\eem{\end{enumerate}}
\def\bex{\begin{example}}
	\def\eex{\end{example}}
\def\bim{\begin{itemize}}
	\def\eim{\end{itemize}}
\def\bl{\begin{lemma}}
	\def\el{\end{lemma}}
\def\bma{\begin{bmatrix}}
	\def\ema{\end{bmatrix}}
\def\bpf{\begin{proof}}
	\def\epf{\end{proof}}
\def\bpp{\begin{proposition}}
	\def\epp{\end{proposition}}
\def\bqu{\begin{question}}
	\def\equ{\end{question}}
\def\br{\begin{remark}}
	\def\er{\end{remark}}
\def\bt{\begin{theorem}}
	\def\et{\end{theorem}}

\def\squareforqed{\hbox{\rlap{$\sqcap$}$\sqcup$}}
\def\qed{\ifmmode\squareforqed\else{\unskip\nobreak\hfil
		\penalty50\hskip1em\null\nobreak\hfil\squareforqed
		\parfillskip=0pt\finalhyphendemerits=0\endgraf}\fi}
\def\endenv{\ifmmode\;\else{\unskip\nobreak\hfil
		\penalty50\hskip1em\null\nobreak\hfil\;
		\parfillskip=0pt\finalhyphendemerits=0\endgraf}\fi}
\newenvironment{proof}{\noindent \textbf{{Proof.~} }}{\qed}
\def\Dbar{\leavevmode\lower.6ex\hbox to 0pt
	{\hskip-.23ex\accent"16\hss}D}
\makeatletter
\def\url@leostyle{%
	\@ifundefined{selectfont}{\def\UrlFont{\sf}}{\def\UrlFont{\small\ttfamily}}}
\makeatother
\def\bcj{\begin{conjecture}}
	\def\ecj{\end{conjecture}}
\def\bcr{\begin{corollary}}
	\def\ecr{\end{corollary}}
\def\bd{\begin{definition}}
	\def\ed{\end{definition}}
\def\bea{\begin{eqnarray}}
	\def\eea{\end{eqnarray}}
\def\bem{\begin{enumerate}}
	\def\eem{\end{enumerate}}
\def\bex{\begin{example}}
	\def\eex{\end{example}}
\def\bim{\begin{itemize}}
	\def\eim{\end{itemize}}
\def\bl{\begin{lemma}}
	\def\el{\end{lemma}}
\def\bpf{\begin{proof}}
	\def\epf{\end{proof}}
\def\bpp{\begin{proposition}}
	\def\epp{\end{proposition}}
\def\bqu{\begin{question}}
	\def\equ{\end{question}}
\def\br{\begin{remark}}
	\def\er{\end{remark}}
\def\bt{\begin{theorem}}
	\def\et{\end{theorem}}

\def\btb{\begin{tabular}}
	\def\etb{\end{tabular}}

	\newcommand{\nc}{\newcommand}
	
	\nc{\bbA}{\mathbb{A}} \nc{\bbB}{\mathbb{B}} \nc{\bbC}{\mathbb{C}}
	\nc{\bbD}{\mathbb{D}} \nc{\bbE}{\mathbb{E}} \nc{\bbF}{\mathbb{F}}
	\nc{\bbG}{\mathbb{G}} \nc{\bbH}{\mathbb{H}} \nc{\bbI}{\mathbb{I}}
	\nc{\bbJ}{\mathbb{J}} \nc{\bbK}{\mathbb{K}} \nc{\bbL}{\mathbb{L}}
	\nc{\bbM}{\mathbb{M}} \nc{\bbN}{\mathbb{N}} \nc{\bbO}{\mathbb{O}}
	\nc{\bbP}{\mathbb{P}} \nc{\bbQ}{\mathbb{Q}} \nc{\bbR}{\mathbb{R}}
	\nc{\bbS}{\mathbb{S}} \nc{\bbT}{\mathbb{T}} \nc{\bbU}{\mathbb{U}}
	\nc{\bbV}{\mathbb{V}} \nc{\bbW}{\mathbb{W}} \nc{\bbX}{\mathbb{X}}
	\nc{\bbZ}{\mathbb{Z}}
	
	\nc{\bA}{{\bf A}} \nc{\bB}{{\bf B}} \nc{\bC}{{\bf C}}
	\nc{\bD}{{\bf D}} \nc{\bE}{{\bf E}} \nc{\bF}{{\bf F}}
	\nc{\bG}{{\bf G}} \nc{\bH}{{\bf H}} \nc{\bI}{{\bf I}}
	\nc{\bJ}{{\bf J}} \nc{\bK}{{\bf K}} \nc{\bL}{{\bf L}}
	\nc{\bM}{{\bf M}} \nc{\bN}{{\bf N}} \nc{\bO}{{\bf O}}
	\nc{\bP}{{\bf P}} \nc{\bQ}{{\bf Q}} \nc{\bR}{{\bf R}}
	\nc{\bS}{{\bf S}} \nc{\bT}{{\bf T}} \nc{\bU}{{\bf U}}
	\nc{\bV}{{\bf V}} \nc{\bW}{{\bf W}} \nc{\bX}{{\bf X}}
	\nc{\ba}{{\bf a}} \nc{\be}{{\bf e}} \nc{\bu}{{\bf u}}
	\nc{\brr}{{\bf r}} \nc{\bx}{{\bf x}}
	
	\nc{\cA}{{\cal A}} \nc{\cB}{{\cal B}} \nc{\cC}{{\cal C}}
	\nc{\cD}{{\cal D}} \nc{\cE}{{\cal E}} \nc{\cF}{{\cal F}}
	\nc{\cG}{{\cal G}} \nc{\cH}{{\cal H}} \nc{\cI}{{\cal I}}
	\nc{\cJ}{{\cal J}} \nc{\cK}{{\cal K}} \nc{\cL}{{\cal L}}
	\nc{\cM}{{\cal M}} \nc{\cN}{{\cal N}} \nc{\cO}{{\cal O}}
	\nc{\cP}{{\cal P}} \nc{\cQ}{{\cal Q}} \nc{\cR}{{\cal R}}
	\nc{\cS}{{\cal S}} \nc{\cT}{{\cal T}} \nc{\cU}{{\cal U}}
	\nc{\cV}{{\cal V}} \nc{\cW}{{\cal W}} \nc{\cX}{{\cal X}}
	\nc{\cZ}{{\cal Z}}
	
	\nc{\hA}{{\hat{A}}} \nc{\hB}{{\hat{B}}} \nc{\hC}{{\hat{C}}}
	\nc{\hD}{{\hat{D}}} \nc{\hE}{{\hat{E}}} \nc{\hF}{{\hat{F}}}
	\nc{\hG}{{\hat{G}}} \nc{\hH}{{\hat{H}}} \nc{\hI}{{\hat{I}}}
	\nc{\hJ}{{\hat{J}}} \nc{\hK}{{\hat{K}}} \nc{\hL}{{\hat{L}}}
	\nc{\hM}{{\hat{M}}} \nc{\hN}{{\hat{N}}} \nc{\hO}{{\hat{O}}}
	\nc{\hP}{{\hat{P}}} \nc{\hR}{{\hat{R}}} \nc{\hS}{{\hat{S}}}
	\nc{\hT}{{\hat{T}}} \nc{\hU}{{\hat{U}}} \nc{\hV}{{\hat{V}}}
	\nc{\hW}{{\hat{W}}} \nc{\hX}{{\hat{X}}} \nc{\hZ}{{\hat{Z}}}
	
	\nc{\hn}{{\hat{n}}}
	
	\def\dim{\mathop{\rm Dim}}

	\def\min{\mathop{\rm min}}

	\newcommand{\bra}[1]{\langle#1|}
	\newcommand{\ket}[1]{|#1\rangle}
	
	\newcommand{\ketbra}[2]{|#1\rangle\!\langle#2|}

	\usepackage[
	colorlinks,
	linkcolor = blue,
	citecolor = blue,
	urlcolor = blue]{hyperref}
	\def \qed {\hfill \vrule height7pt width 7pt depth 0pt}
	
	\newcounter{lastnote}

\begin{document}
		\title{Almost all even-particle pure states are determined by their half-body marginals}
		
		\author{Wanchen Zhang}
		\affiliation{School of Mathematical Sciences,
			University of Science and Technology of China, Hefei, 230026,  China}
		\affiliation{Hefei National Laboratory, University of Science and Technology of China, Hefei, 230088, China}

		\author{Fei Shi}
		\affiliation{QICI Quantum Information and Computation Initiative, Department of Computer Science,
			The University of Hong Kong, Pokfulam Road, Hong Kong}	
		
		\author{Xiande Zhang}
		\email[]{Corresponding author: drzhangx@ustc.edu.cn}
		\affiliation{School of Mathematical Sciences,
			University of Science and Technology of China, Hefei, 230026,  China}
		\affiliation{Hefei National Laboratory, University of Science and Technology of China, Hefei, 230088, China}

		\begin{abstract}
			Determining whether the original global state is uniquely determined by its local marginals is a prerequisite for some efficient tools for characterizing quantum states. 
			In this paper, we show that almost all generic pure states
			of even $N$-particle with equal local dimension are \emph{uniquely determined among all other pure states} (UDP) by four of  their half-body marginals.
			Furthermore, we give a graphical description of the marginals  for determining  genuinely multipartite entangled states, which leads to several lower bounds on the number of required marginals.
		    Finally, we present a construction of $N$-qudit states obtained from certain combinatorial structures that cannot be UDP by its $k$-body marginals for some $k\geq \lfloor N/2\rfloor$.
		\end{abstract}
		\maketitle
		\vspace{-0.5cm}

		\section{Introduction}\label{sec:int}
		Understanding the property of multi-particle systems with only a few particles is one of the central problems in quantum physics. For a given multi-particle quantum state, calculating its \emph{reduced density matrices} (RDMs), which we call marginal in the following, is straightforward. In the opposite direction, a natural mathematical problem arises: For given marginals of some subsystems, are they compatible with each other in the sense that they can arise from a global state? The task to determine all compatible tuples of marginals and to describe this set in an elegant way is called \emph{quantum marginal problem} (QMP) \cite{Klyachko_2006,schilling2015quantum,haapasalo2021quantum,Schilling2017Quantum}.
		The QMP has a long history of research and has played an important role in quantum chemistry \cite{QMP,RevModPhys.35.668}. An improtant question under this problem is: For a given set of marginals, whether the original global state the only state having the given set of marginals. Recently, there has been increasing interest in quantum
		state tomography via marginals, see for example \cite{PhysRevLett.118.020401,PhysRevLett.124.100401}.
		A fundamental problem in adopting tomography is to determine what original states can be uniquely determined by its marginals.
		
		Many researchers have considered this question over the past two decades. Linden $\emph{et al.}$ showed that almost all pure states of three-qubit are  \emph{uniquely determined among arbitrary states (mixed or pure)} (UDA) by their 2-body marginals, and the only states without this property are equivalent under local rotations to
		states of the form $a\ket{000}+b\ket{111}$ \cite{PhysRevLett.89.207901}. Actually, an $N$-qubit pure state is undetermined by its marginals among arbitrary states if and only if it is local unitary equivalent to a generalized $N$-qubit \emph{Greenberger-Horne-Zeilinger} (GHZ) state \cite{PhysRevLett.100.050501,PhysRevA.79.032326}. Later, Jones and Linden showed that almost all pure states of $N$-qudit are UDA by its $\lceil N/2+1 \rceil$-body marginals, and the set of all $\lfloor N/2 \rfloor$-body marginals is an insufficient description of an $N$-qudit pure quantum state among arbitrary states \cite{PhysRevA.71.012324}. Recently, Huang $\emph{et al.}$ proved that
		almost all odd $N$-qudit pure states are UDA by its $(N+1)/2$-body marginals \cite{huang2018quantum}, thus completing the determination of threshold on the number of bodies for the UDA problem.
		
		The QMP can be viewed as a quantum analog of  \emph{Ulam's reconstruction conjecture}, which arises in graph theory \cite{Huber_2018}. The conjecture can be formulated as follows: given a complete set of single-vertex-deleted subgraphs, is the joint graph (without vertex labels) uniquely determined? Although Bollob\'{a}s proved that almost all graphs can be uniquely determined by three vertex deleted subgraphs \cite{Almostgraphthree} in 1990, this remains one of the outstanding open questions in graph theory. In particular, reconstructing the original graph with fewer vertex-deleted subgraphs  is important in \emph{Ulam's reconstruction problem}.
		
		Similarly in QMP, researchers consider reconstructing the original quantum state with as few marginals as possible. However, results on this theme were focused on the set of pure states due to complexity.  Di\`{o}si proved that almost all three-particle pure states, which do not require the same local dimensions, are \emph{uniquely determined among all other pure states} (UDP) by just two 2-body marginals \cite{PhysRevA.70.010302}.
		This result can be extended to that almost all $N$-particle pure states can be UDP by two intersecting marginals \cite{2012Comment}.
		Recently, Wyderka $\emph{et al.}$ showed that almost all four-qudit pure states are UDP by three specified 2-body marginals \cite{PhysRevA.96.010102}, and they extended this result to that almost all $N$-qudit pure states are UDP by three specified $(N-2)$-body marginals for any $N\ge 4$. 
		The word ``almost all'' implies two points, one is that states are drawn randomly according to the Haar measure, which has full Schmidt rank and pairwise distinct Schmidt coefficients \cite{PhysRevA.96.010102}, and the other is that all states except some parameter-specific ones, which may bring counterexamples \cite{PhysRevLett.89.207901}. For example, we know that neither the generalized GHZ states nor their local unitary equivalence classes can be UDP by their marginals \cite{PhysRevLett.100.050501}.		

		In this paper, we explore more results on this direction.
\begin{itemize}
  \item Our first contribution is to show that for any even $N\ge 4$, almost all $N$-qudit generic pure states are UDP by two specified pairs of half-body marginals. This result is interesting in two aspects.  On the one hand, for even $N$, $N/2$ is not sufficient for UDA \cite{PhysRevA.71.012324}, but it is sufficient for UDP based on our result. On the other hand, using half-body marginals, that is $N/2$-body, might be the best-known result for the general UDP problem (comparing with $(N/2+1)$-body in \cite{2012Comment}); further, only four specified half-body marginals is proved to be enough.
  \item Our second contribution is, for genuinely multipartite entangled (GME) states, we establish a necessary condition for the set of marginals to determine GME states by using hypergraphs and giving several lower bounds on the number of marginals required.
  \item Finally, we investigate states that can not be UDP. Our third contribution is to present a construction of $N$-qudit states obtained from orthogonal arrays (OAs) or packing arrays (PAs) that cannot be UDP by its $k$-body marginals for some $k\geq \lfloor N/2\rfloor$.
\end{itemize} 

		The rest of this paper is organized as follows. In Sec.~\ref{sec:pre}, we introduce preliminary knowledge. Sections~\ref{sec:qubit}-\ref{sec:counterexample} are devoted to obtaining our three results listed above one in a section. Finally, we conclude in Sec.~\ref{sec:con}.

		\section{Preliminaries}\label{sec:pre}
		In this section, we introduce the preliminary knowledge and facts.
		For a given set of indices $\mathcal{J}=\{j_1, j_2,\ldots, j_k\}$, let $\mathcal{J}_C$ denote the complement of  $\mathcal{J}$ in $\mathcal{I}=[N]:=\{1, 2, \ldots, N\}$, and let $P(\mathcal{I})$ denote the power set of $\mathcal{I}$.
		
		
		\begin{definition}
			Given an $N$-particle quantum state $\rho$ of parties $\mathcal{I}$ and a $k$-subset $\mathcal{J}\subset \mathcal{I}$, its $k$-body marginal of parties $\mathcal{J}$ is defined as
			\begin{equation}
				\rho_\mathcal{J}\triangleq \text{Tr}_{\mathcal{J}_C}(\rho),
			\end{equation}
			where the trace is a partial trace over parties $\mathcal{J}_C=\mathcal{I} \setminus \mathcal{J}$.
		\end{definition}
		
		Given a family $\mathcal{F}\subset P(\mathcal{I})$, we denote
		\begin{equation}
			\mathcal{D}_{\mathcal{F}}(\rho):=\{\rho_{\mathcal{J}}: \mathcal{J}\subseteq \mathcal{F}\},
		\end{equation}
		and call this collection of quantum marginals as a \emph{deck by $\mathcal{F}$}.
		If $\mathcal{F}$ consists of only $k$-subsets, that is, the deck contains only $k$-body marginals, we say it is a $k$-deck.
		Particularly,
		a $k$-deck is called complete if it contains all the $\binom{N}{k}$ $k$-body marginals, and denoted by $\mathcal{D}_k(\rho)$. When $\rho=\ketbra{\psi}{\psi}$, which is a pure state, we also use the notation $\mathcal{D}_k(\ket{\psi})$ for $\mathcal{D}_k(\rho)$.

		\begin{definition}\label{def-udp}
		 	A pure state $\ket{\psi}$   of parties $\mathcal{I}$ is called
\begin{itemize}
  \item \emph{uniquely determined by $\mathcal{F}$ among arbitrary states} ($\mathcal{F}$-UDA)
			if there exists no other state $\rho'$  of parties $\mathcal{I}$ satisfying  $\mathcal{D}_{\mathcal{F}}(\ket{\psi})=\mathcal{D}_{\mathcal{F}}(\rho')$.
  \item 	\emph{uniquely determined by $\mathcal{F}$ among pure states} ($\mathcal{F}$-UDP)
			if there exists no other pure state $\ket{\psi'}$  of parties $\mathcal{I}$ satisfying  $\mathcal{D}_{\mathcal{F}}(\ket{\psi})=\mathcal{D}_{\mathcal{F}}(\ket{\psi'})$.
\end{itemize}
			 Further, if each subset in $\mathcal{F}$ has size at most $k$, then we say $\ket{\psi}$ is $k$-UDA or $k$-UDP.
		\end{definition}
		
		Using Definition~\ref{def-udp}, Di\`{o}si's result indicates that almost all three-particle pure states are $\mathcal{F}$-UDP if $\mathcal{F}$ consists of any two subsets of size two \cite{PhysRevA.70.010302}; almost all $N$-particle pure states are $\mathcal{F}$-UDP if $\mathcal{F}$ consists of any two intersecting subsets $P_1, P_2$ with $P_1\cup P_2= \mathcal{I}$  \cite{2012Comment}; almost all four-qudit pure states are $\mathcal{F}$-UDP if $\mathcal{F}$ consists of $AB$, $BC$ and $CD$ \cite{PhysRevA.96.010102}.  When considering the smallest $k$ such that $\ket{\psi}$ is $k$-UDP, for almost all pure states, we conclude that $k\leq \lceil N/2 \rceil+1$ by \cite{2012Comment}, and $k\leq 2=4/2$ for four-qudit pure states by \cite{PhysRevA.96.010102}. The first contribution in our paper indeed shows that the smallest $k$ such that almost all $\ket{\psi}$ is $k$-UDP satisfies $k\leq N/2 $ when $N$ is even, which improves the results in \cite{2012Comment} and generalizes the result in \cite{PhysRevA.96.010102}. Combining the result in \cite{PhysRevA.71.012324}, we conclude that, for even $N$, almost all $N$-qudit pure states are not $N/2$-UDA, but are $N/2$-UDP.
		

		
		
		Now, we give the definition of generic states introduced in \cite{PhysRevA.96.010102}, which can be viewed as random states selected according to the Haar measure. For even $N$, consider a pure state $\ket{\psi}\in\mathcal{H}_{1}\otimes \cdots\otimes \mathcal{H}_{N}\cong (\mathbb{C}^d)^{\otimes N}$. Using the Schmidt decomposition along the bipartition $A_{1}\ldots A_{{N/2}}|A_{{N/2+1}}\ldots A_{{N}}$, we can write the state as
		\begin{equation}\label{generic}
			\ket{\psi}=\sum_{i=1}^{d^{N/2}}\sqrt{\lambda_i}\ket{i}_{A_{1}\ldots A_{{N/2}}}\otimes\ket{i}_{A_{{N/2+1}}\ldots A_{{N}}}
		\end{equation}
		where $\sum_{i=1}^{d^{N/2}}\lambda_i=1$.  If the state has full Schmidt rank,
		i.e., $\lambda_i\neq0$ for all $i$, then the two sets of states $\{\ket{i}_{A_{1}\ldots A_{{N/2}}}\}_{i=1}^{d^{N/2}}$ and $\{\ket{i}_{A_{{N/2+1}}\ldots A_{{N}}}\}_{i=1}^{d^{N/2}}$ are
		orthonormal bases of the composite Hilbert spaces $\mathcal{H}_{A_1}\otimes\cdots\otimes\mathcal{H}_{A_{N/2}}$ and $\mathcal{H}_{A_{N/2+1}}\otimes\cdots\otimes\mathcal{H}_{A_N}$, respectively.

		\begin{definition}\cite{PhysRevA.96.010102}
			A generic even N-particle pure state is a
			state $\ket{\psi}\in\mathcal{H}_{1}\otimes \cdots\otimes \mathcal{H}_{N}\cong (\mathbb{C}^d)^{\otimes N}$ drawn randomly according to the Haar measure. Writing such a state as in Eq. (\ref{generic}), the Schmidt bases and the set of Schmidt coefficients are independent of each other. The distribution
			of the Schmidt coefficients is given by \cite{LLOYD1988186,AJScott_2003,PhysRevA.96.010102}
			\begin{equation}
				\begin{aligned}
					&P(\lambda_1, \ldots, \lambda_{d^{N/2}})d\lambda_1\ldots\lambda_{d^{N/2}}=\\
					&N\delta(1-\sum_{i=1}^{d^{N/2}}\lambda_i)\prod\limits_{1\le i < j\le d^{N/2}} (\lambda_i-\lambda_j)^2d\lambda_i\ldots d\lambda_{d^{N/2}}
				\end{aligned}
			\end{equation}
			and the Schmidt bases are distributed according to the Haar measure of unitary operators on the smaller Hilbert spaces.
		\end{definition}
		
		Mathematically, a generic pure state means that when this state is formulated in the form of Eq. (\ref{generic}), it has full Schmidt rank and pairwise distinct Schmidt coefficients.
		In the next section, we will show that for even $N$, almost all $N$-particle generic pure states can be UDP by four of its half-body marginals.

		\section{Almost all even-particle pure state are determined by their half-body marginals}\label{sec:qubit}
		In this section, we show that for all even $N\geq 4$, almost all $N$-qudit pure states are determined by their half-body marginals, that is, almost all $N$-qudit pure states are $N/2$-UDP.

		When $N\equiv0 \mod 4$, we can apply  \cite[Theorem~2]{PhysRevA.96.010102} to obtain the result that almost all $N$-qudit pure states are determined by three specific half-body marginals, i.e., $N/2$-UDP. Here, we only need to divide the $N$ particles into four parts equally and recompose each part as a particle. 
		
		When $N\equiv2 \mod 4$, \cite[Theorem~2]{PhysRevA.96.010102} is not applicable. However, by \cite{2012Comment}, we know that two intersecting $(N/2+1)$-body marginals can determine the original pure state, that is, almost all $N$-qudit pure states are $(N/2+1)$-UDP in this case. Motivated by the proof of \cite[Theorem~1]{PhysRevA.96.010102}, we can improve this result and show that almost all $N$-qudit pure states are $N/2$-UDP by four of its half-body marginals.

We state our main result in this section as follows. Denote $N=2n$ for convenience.  For a subset of $[N]$, we also consider it as a set of particles, that is, we view $A\subset[N]$ as $\{\mathcal{H}_i:i\in A\}$. We use $AB$ to denote $A \cup B$ when $A \cap B = \emptyset$.

		\begin{theorem}\label{theorem111}
			Almost all $2n$-qudit pure states $\ket{\psi}$ are UDP by its $\mathcal{D}_n(\ket{\psi})$ for $n\ge 2$. In fact, any four $n$-body marginals $\rho_{AB}, \rho_{CD}$ and $\rho_{AC}, \rho_{BD}$ satisfying that $(AB|CD)$ and  $(AC|BD)$ are both bipartitions  of the set of  $2n$ particles will suffice (see Fig.~\ref{figeven}).
		\begin{figure*}[htbp]
			\centering
			\includegraphics[width=0.8 \textwidth]{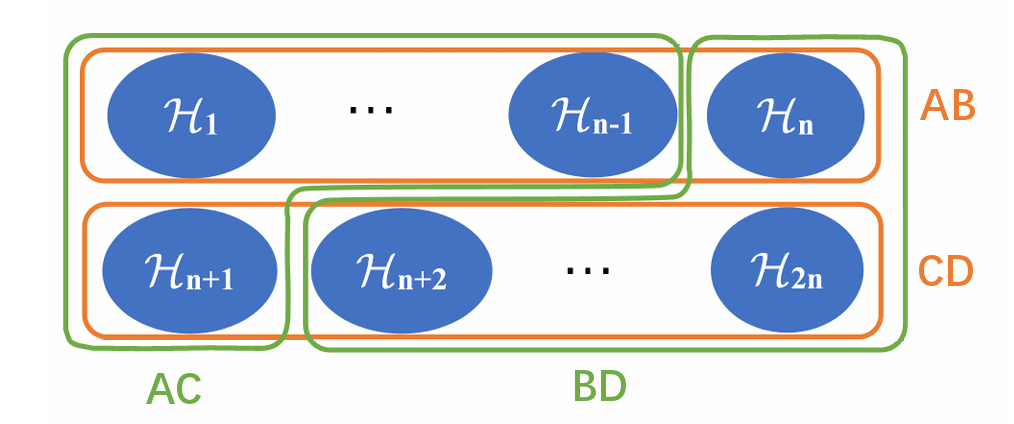}
			\caption{ Illustration of the worst case of $2n$-qudit in Theorem~\ref{theorem111}: the case where these two pairs of bipartitions yield the least number of equations to determine the state, which happens when these two bipartitions are the most unbalanced.}\label{figeven}
		\end{figure*}		
		\end{theorem}

The proof of Theorem~\ref{theorem111} is deferred to Section~\ref{sec-six}. The idea of the proof follows from the one in \cite[Theorem~1]{PhysRevA.96.010102}, 
where the authors proved that three specific two-body marginals $\rho_{AB}, \rho_{CD}$ and $\rho_{BD}$ are enough for almost all four-qubit pure states to be UDP among pure states of four particles $A$, $B$, $C$ and $D$.
The rough idea of their argument is that one first looks at the
bipartition $\rho_{AB}, \rho_{CD}$, and generically the Schmidt spectrum is nondegenerate;
This implies that the only remaining freedom can be parameterized by
some phases; One then uses the other marginal $\rho_{BD}$ and shows that there
are sufficiently many equations to constrain the phases. However, only one marginal besides the bipartition is not enough, one needs to discover more constraints on the phases, see Eqs. (28)-(32) in \cite{PhysRevA.96.010102}, which are not easy to be generalized to $N$ qudits.

To get a similar result for general even $N$, we use another bipartition  $\rho_{AC}, \rho_{BD}$ in Theorem~\ref{theorem111} to obtain enough independent equations. The result is weaker for the case $N=4$, but it is executable for the proof of general even $N$. For easy reading, we first deal with the case of six qubits as an example in Section~\ref{sec-six} and then prove  Theorem~\ref{theorem111}.
 Before moving to Section~\ref{sec-six}, we need to introduce several notations used in our proof, most of which follow from those in  \cite{PhysRevA.96.010102}.

%
Suppose that $\ket{\psi}$ is a generic pure state of $N=2n$ particles of internal dimension $d$, then the Schmidt decomposition along the bipartition $(AB|CD)$ gives
        \begin{equation}\label{eqpsi2N}
        	\ket{\psi}= \sum_{i=1}^{d^n}\sqrt{\lambda_i}\ket{i}_{AB}\ket{i}_{CD},
        \end{equation}
  where $\sqrt{\lambda_i}$ are the strictly decreasingly ordered nonzero Schmidt coefficients.   It is clear that  $\rho_{AB}=\sum_{i=1}^{d^n}{\lambda_i}\ket{i}\bra{i}_{AB}$ and $\rho_{CD}=\sum_{i=1}^{d^n}{\lambda_i}\ket{i}\bra{i}_{CD}$. If another pure state $\ket{\phi}$ shares the same half-body marginals $\rho_{AB}$ and $\rho_{CD}$ as $\ket{\psi}$, then it must be of the form
	
		\begin{equation}\label{eqphi2N}
     	\ket{\phi}= \sum_{i=1}^{d^n}e^{\sqrt{-1}\varphi_i}\sqrt{\lambda_i}\ket{i}_{AB}\ket{i}_{CD}.
         \end{equation}
		

		We want to show that all $\varphi_i$ are the same, thus $\ket{\psi}=\ket{\phi}$. In Theorem~\ref{theorem111},
		 we demand that another bipartition of marginals $\rho_{AC}, \rho_{BD}$ coincide, i.e. $\text{Tr}_{BD}(\ket{\psi}\bra{\psi})=\text{Tr}_{BD}(\ket{\phi}\bra{\phi})$ and $\text{Tr}_{AC}(\ket{\psi}\bra{\psi})=\text{Tr}_{AC}(\ket{\phi}\bra{\phi})$, thus for each $\mathcal{J}\in \{BD, AC\}$,
\begin{equation}\label{rhobd}
			\begin{aligned}
				\sum_{i,j=1}^{d^n}\sqrt{\lambda_i\lambda_j}\text{Tr}_{\mathcal{J}}(\ket{i}\bra{j}_{AB}\otimes\ket{i}\bra{j}_{CD})
				&= \text{Tr}_{\mathcal{J}}(\ket{\psi}\bra{\psi}) \\
				&= \text{Tr}_{\mathcal{J}}(\ket{\phi}\bra{\phi}) \\
				&= \sum_{i,j=1}^{d^n}e^{\sqrt{-1}(\varphi_i-\varphi_j)}\sqrt{\lambda_i\lambda_j}\text{Tr}_{\mathcal{J}}(\ket{i}\bra{j}_{AB}\otimes\ket{i}\bra{j}_{CD}).
			\end{aligned}
		\end{equation}
Denote
\begin{equation}\label{eqgamma}
  \gamma^{(ij)}:=(1-e^{\sqrt{-1}(\varphi_i-\varphi_j)})\sqrt{\lambda_i\lambda_j}, \text{ for } i,j=1,\ldots, d^n.
\end{equation} Then we have
\begin{equation}\label{eqsum}
         \sum_{i,j=1}^{d^n}\gamma^{(ij)}\text{Tr}_{\mathcal{J}}(\ket{i}\bra{j}_{AB}\otimes\ket{i}\bra{j}_{CD})=0.
        \end{equation}
Note that showing that all $\varphi_i$ are the same is equivalent to showing that $\gamma^{(ij)}$ has only zero solutions. Since $\gamma^{(ii)}=0$ and $\gamma^{(ij)}=\overline\gamma^{(ji)}$, there are in fact $d^n(d^n-1)/2$ undetermined complex-valued variables $\gamma^{(ij)}$.
The main task is to discover enough independent equations from Eq. (\ref{eqsum}) to show that $\gamma^{(ij)}=0$  for all $i,j$. To that end, we further define some notations that will be used in the computation process as in  \cite{PhysRevA.96.010102}.
		
		When $\mathcal{J}=BD$, note that $\text{Tr}_{BD}(\ket{i}\bra{j}_{AB}\otimes\ket{i}\bra{j}_{CD})=\text{Tr}_{B}(\ket{i}\bra{j}_{AB})\otimes \text{Tr}_{D}(\ket{i}\bra{j}_{CD})$.
		Denote
 \begin{equation}\label{qp}
Q^{(ij)}:=\text{Tr}_{B}(\ket{i}\bra{j}_{AB}) \text{ and } P^{(ij)}:=\text{Tr}_{D}(\ket{i}\bra{j}_{CD}).
\end{equation}
Then  $\text{Tr}_{BD}(\ket{i}\bra{j}_{AB}\otimes\ket{i}\bra{j}_{CD})= Q^{(ij)} \otimes P^{(ij)}$. Further,
 \begin{equation}\label{qcong}
\text{Tr}(Q^{(ij)})=\delta_{ij}, Q^{(ij)\dagger}=Q^{(ji)}, \text{ and }\text{Tr}(P^{(ij)})=\delta_{ij}, P^{(ij)\dagger}=P^{(ji)}.
\end{equation}
So Eq. (\ref{eqsum}) can be written as
 \begin{equation}\label{eqre}
\sum_{i<j}\gamma^{(ij)} Q^{(ij)}\otimes P^{(ij)}+\overline\gamma^{(ij)}Q^{(ij)\dagger}\otimes P^{(ij)\dagger}=0
\end{equation}
 Similarly for $\mathcal{J}=AC$, denoting
\begin{equation}\label{lm}
L^{(ij)}:=\text{Tr}_{A}(\ket{i}\bra{j}_{AB}) \text{ and } M^{(ij)}:=\text{Tr}_{C}(\ket{i}\bra{j}_{CD}),\end{equation}
 we have $\text{Tr}_{AC}(\ket{i}\bra{j}_{AB}\otimes\ket{i}\bra{j}_{CD})= L^{(ij)} \otimes M^{(ij)}$, and
 \begin{equation}\label{eqre1}
 \sum_{i<j}\gamma^{(ij)} L^{(ij)}\otimes M^{(ij)}+\overline\gamma^{(ij)}L^{(ij)\dagger}\otimes M^{(ij)\dagger}=0
 \end{equation}

Next, we treat Eqs. (\ref{eqre}) and (\ref{eqre1}) as two sets of equations by considering each entry of the matrices. Let $q^{(ij)}_{ab}$, $p^{(ij)}_{ab}$, $l^{(ij)}_{ab}$ and $m^{(ij)}_{ab}$ be the element in the $(a,b)$-entry of $Q^{(ij)}$, $P^{(ij)}$, $L^{(ij)}$ and $M^{(ij)}$, respectively.
Similar to  \cite{PhysRevA.96.010102}, we need to analyze the dependence among these numbers to ensure that the equations used later in the proof are independent of each other. Thus we need the following lemma.

		
		\begin{lemma}\label{dependence}
			For any $i\neq j$, the only dependence among  the number in $\{q_{ab}^{(ij)},l_{ab}^{(ij)},p_{ab}^{(ij)},m_{ab}^{(ij)}\}_{a,b}$ are $\sum_cq^{(ij)}_{cc}=\sum_cl^{(ij)}_{cc}=\sum_cp^{(ij)}_{cc}=\sum_cm^{(ij)}_{cc}=0$.
		\end{lemma}
	\begin{proof} It suffices to show that among $E\triangleq\{q_{ab}^{(ij)},l_{ab}^{(ij)}\}_{a,b}$, the only dependence is $\sum_cq^{(ij)}_{cc}=\sum_cl^{(ij)}_{cc}=0$. If this is true, then similarly the only dependence among $\{p_{ab}^{(ij)},m_{ab}^{(ij)}\}_{a,b}$ is $\sum_cp^{(ij)}_{cc}=\sum_cm^{(ij)}_{cc}=0$. Since $AB$ and $CD$ are disjoint sets of particles, and their  Schmidt  bases are chosen independently and randomly, we conclude the result.

Now we expand the Schmidt bases $\ket{i}_{AB}$ in terms of the computational basis by $\ket{i}_{AB}=\sum_{a,b}\mu^{(i)}_{ab}\ket{ab},$ 
        where the only dependence among the $\mu^{(i)}_{ab}$ is
        \begin{equation}\label{dep}
        	\langle{i}\ket{j}_{AB}=\sum_{a,b}\mu^{(i)}_{ab}\overline\mu^{(j)}_{ab}=\delta_{ij}.
        \end{equation}
         By Eqs. (\ref{qp}) and (\ref{lm}), we can
       express the numbers $q_{ef}^{(ij)}$ and $l_{ef}^{(ij)}$  in terms of the coefficients $\mu$ as follows,
        \begin{equation}\label{qef}
        	q_{ef}^{(ij)}=\sum_b\mu^{(i)}_{eb}\overline\mu^{(j)}_{fb}, ~ l_{ef}^{(ij)}=\sum_{a}\mu^{(i)}_{ae}\overline\mu^{(j)}_{af}.
        \end{equation}
%
By Eq. (\ref{dep}), for $i\neq j$, we have $\sum_cq^{(ij)}_{cc}=\sum_cl^{(ij)}_{cc}=0$ as in Eq. (\ref{qcong}). Next, we show these are the only dependencies among those numbers in $E$.

Consider a number $q_{ef}^{(ij)}$ with $e\neq f$, which contains an item $\mu^{(i)}_{eb}\overline\mu^{(j)}_{fb}$  with $e\neq f$ for arbitrary $b$ by Eq.~(\ref{qef}). Since the item $\mu^{(i)}_{eb}\overline\mu^{(j)}_{fb}$ does not appear in the summation in Eq.~(\ref{dep}), it can not be represented by any other $\mu^{(i)}_{ab}\overline\mu^{(j)}_{a'b'}$.
Further by Eq.~(\ref{qef}), the item $\mu^{(i)}_{eb}\overline\mu^{(j)}_{fb}$  does not appear in the  numbers from $E$ besides $q_{ef}^{(ij)}$, which means we cannot use numbers from $E$ to represent $\mu^{(i)}_{eb}\overline\mu^{(j)}_{fb}$ without $q_{ef}^{(ij)}$. Consequently, $q_{ef}^{(ij)}$ cannot be represented by any other elements in $E$. Similar arguments works for $l_{ef}^{(ij)}$ with $e\neq f$. Thus we obtain that the  numbers in $\{q_{ef}^{(ij)},l_{ef}^{(ij)}\}_{e\neq f}$ are linearly independent,  and they are linearly independent of numbers in $\{q_{cc}^{(ij)},l_{cc}^{(ij)}\}_c$.

		
		
It is left to prove that the only dependence among $\{q_{cc}^{(ij)},l_{cc}^{(ij)}\}_c$ is $\sum_cq^{(ij)}_{cc}=\sum_cl^{(ij)}_{cc}=0$. Let $E_1=\{q_{cc}^{(ij)},l_{cc}^{(ij)}\}_{c>1}$, that is, we exclude the numbers $q^{(ij)}_{11}$ and $l^{(ij)}_{11}$. Then it is equivalent to show that numbers in $E_1$ are independent of each other. Consider the number $q^{(ij)}_{ee}$ with $e>1$, which contains an item $\mu^{(i)}_{e1}\overline\mu^{(j)}_{e1}$. It is easy to check that the item $\mu^{(i)}_{e1}\overline\mu^{(j)}_{e1}$ does not appear in any other number in $E_1$. This means $q^{(ij)}_{ee}$ cannot be represented by other elements in $E_1$, i.e., independent of other numbers in $E_1$. Similarly, we have $l^{(ij)}_{ee}$ is independent of other numbers in $E_1$, which completes the proof.
%
%
%
	\end{proof}
	%
			\subsection{Proof of Theorem~\ref{theorem111}}\label{sec-six}	
       We first prove that Theorem~\ref{theorem111} is true for the six qubit case as an example. For this case, Fig.~\ref{figeven} is reduced to the unique one in Fig.~\ref{fig2}. The following lemma is stated  by treating $\mathcal{H}_1\otimes\mathcal{H}_2$ as $A$, $\mathcal{H}_3$ as $B$, $\mathcal{H}_4$ as $C$ and $\mathcal{H}_5\otimes\mathcal{H}_6$ as $D$. The first part of the proof is a generalization of \cite[Theorem~1]{PhysRevA.96.010102}, and the second part  is an application of Lemma~\ref{dependence}. For self-containment, we give all the details of the proof.
	
\begin{figure}[H]
	\centering
	\includegraphics[width=0.4 \textwidth]{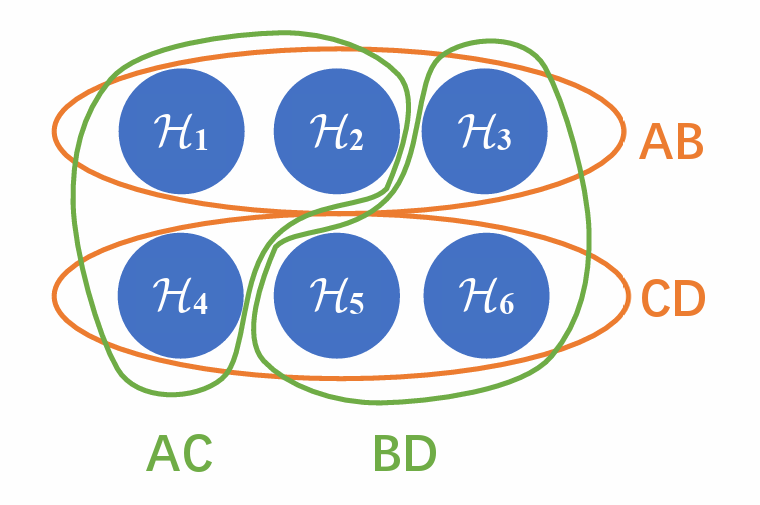}
	\caption{Case of six-qubit: a set of
		four $3$-body marginals that are shown to be sufficient to uniquely
		determine pure generic states.}\label{fig2}
\end{figure}
	
		\begin{lemma}\label{lemma2}
			Almost all six-qubit pure states are UDP by the four $3$-body marginals $\rho_{123}$, $\rho_{456}$, $\rho_{124}$ and $\rho_{356}$.
		\end{lemma}
	
	\begin{proof}
		By adopting the Schmidt decomposition between $AB=\{1, 2, 3\}$ and $CD=\{4, 5, 6\}$, a generic six-qubit pure state can be
		written as
		\begin{equation}
			\ket{\psi}= \sum_{i=1}^{8}\sqrt{\lambda_i}\ket{i}_{\mathcal{H}_1\otimes \mathcal{H}_2\otimes\mathcal{H}_3}\ket{i}_{\mathcal{H}_4\otimes \mathcal{H}_5\otimes\mathcal{H}_6},
		\end{equation}
		where $\sum_i{\lambda_i}=1$. Suppose that there is another pure state $\ket{\phi}$ which exhibits the same three-body marginals as $\ket{\psi}$. Then $\ket{\phi}$ must be of the form
		\begin{equation}
			\ket{\phi}= \sum_{i=1}^{8}e^{\sqrt{-1}\varphi_i}\sqrt{\lambda_i}\ket{i}_{\mathcal{H}_1\otimes \mathcal{H}_2\otimes\mathcal{H}_3}\ket{i}_{\mathcal{H}_4\otimes \mathcal{H}_5\otimes\mathcal{H}_6}.
		\end{equation}
For each $i,j=1,\ldots,8$, let  $\gamma^{(ij)}$ be defined as in Eq. (\ref{eqsum}), where there are indeed $2^3(2^3-1)/2=28$ undetermined complex-valued variables $\gamma^{(ij)}$. Let $Q^{(ij)}$ and $M^{(ij)}$ be $4\times 4$ matrices, $P^{(ij)}$ and $L^{(ij)}$ be $2\times 2$ matrices, as defined in Eqs.  (\ref{qp}) and (\ref{lm}). Considering $Q^{(ij)}\otimes P^{(ij)}$ as a $4\times 4$ block matrix with each block $q_{ab}^{(ij)}P^{(ij)}$, then Eq. (\ref{eqre}) implies that
 \begin{equation}\label{eqf}
\sum_{i<j}\gamma^{(ij)} q^{(ij)}_{ab}P^{(ij)}+\overline\gamma^{(ij)}\overline q^{(ij)}_{ba} P^{(ij)\dagger}=0, \text{ for } a,b=1,\ldots,4.
\end{equation}
For each $1\leq a<b\leq 4$ in Eq. (\ref{eqf}), there are three homogeneous linear equations among the four for each entry due to $\text{Tr}(P^{(ij)})=0$. So there are at least  $6\times 3=18$ homogeneous linear equations from Eq. (\ref{eqf}).

Similarly,  considering $L^{(ij)}\otimes M^{(ij)}$ as a $2\times 2$ block matrix with each block $l_{ab}^{(ij)}M^{(ij)}$, then Eq. (\ref{eqre1}) implies that
 \begin{equation}\label{eqf1}
\sum_{i<j}\gamma^{(ij)} l^{(ij)}_{ab}M^{(ij)}+\overline\gamma^{(ij)}\overline l^{(ij)}_{ba} M^{(ij)\dagger}=0, \text{ for } a,b=1,2.
\end{equation}
For $(a,b)=(1,2)$ in Eq. (\ref{eqf1}), there are $15$ homogeneous linear equations among the $4\times 4=16$ due to $\text{Tr}(M^{(ij)})=0$.

Note that in these $18+15=33$ homogeneous linear equations obtained from Eqs. (\ref{eqf}) and (\ref{eqf1}), we do not involve $q^{(ij)}_{cc}$ and $l^{(ij)}_{cc}$, and we have taken $\text{Tr}(P^{(ij)})=\text{Tr}(M^{(ij)})=0$ into account. So by Lemma~\ref{dependence}, these $33$ homogeneous linear equations are independent for the only $28$ undetermined complex-valued variables $\gamma^{(ij)}$.  Consequently,   $\gamma^{(ij)}=0$ for all $i,j=1,\ldots,8$ and all phases $\varphi_i=\varphi$ must be equal. Thus $\ket{\phi}=e^{\sqrt{-1}\varphi}\ket{\psi}$ which corresponds to the same physical state.
	\end{proof}
	
	Note that there could be more independent linear equations obtained from  Eqs. (\ref{eqf}) and (\ref{eqf1}), but the current selection is sufficient for our proof.

\vspace{0.4cm}

	Next we extend the proof of Lemma~\ref{lemma2} to the case of general even $N$ qudit, that is Theorem~\ref{theorem111}.
		
%
%
%

		\begin{proof} The proof follows the same steps as in  Lemma \ref{lemma2}, so we only indicate the difference. Without loss of generality, we can fix a bipartition $AB|CD$, and let the other bipartition $AC|BD$ be flexible as stated in Theorem~\ref{theorem111}. Using the Schmidt decomposition along the bipartition $AB|CD$, we get $\frac{1}{2}d^{n}(d^{n}-1)$ undetermined complex-valued variables $\gamma^{(ij)}$ for $1 \le i < j \le d^{n}$. The different choices of  a subset $A$ from $AB$ and a subset $C$ from $CD$ will affect the sizes of matrices $Q$, $P$, $L$, and $M$. And the sizes of these matrices will affect the number of independent homogeneous linear equations obtained from Eqs. (\ref{eqf}) and (\ref{eqf1}). In fact, the size of $Q$ is $d^{|A|}\times d^{|A|}$; the size of $P$ is $d^{|C|}\times d^{|C|}$; the size of $L$ is $d^{|B|}\times d^{|B|}$; and the size of $M$ is $d^{|D|}\times d^{|D|}$. Following the proof of  Lemma \ref{lemma2}, Eq. (\ref{eqf}) creates $\binom{d^{|A|}}{2}\times (d^{2|C|}-1)$ linear equations, and Eq. (\ref{eqf1})  creates $\binom{d^{|B|}}{2}\times (d^{2|D|}-1)$ linear equations. By Lemma~\ref{dependence}, these equations are linearly independent. Since $AB|CD$ and $AC|BD$ are both bipartition,   the number of independent equations $\binom{d^{|A|}}{2}\times (d^{2|C|}-1)+\binom{d^{|B|}}{2}\times (d^{2|D|}-1)=\binom{d^{|A|}}{2}\times (d^{2(n-|A|)}-1)+\binom{d^{n-|A|}}{2}\times (d^{2|A|}-1)$, which is a one-variable function. By some computations, we find that this function takes its minimum value at $|A|=n-1$ or $|A|=1$. That is, the most unbalanced case as in Fig.~\ref{figeven} provides the least number of independent equations.
%
			 It can be computed that in this worst case, $\frac{1}{2}d^{n-1}(d^2-1)(d^{n-1}-1)+\frac{1}{2}d(d-1)(d^{2n-2}-1)-\frac{1}{2}d^{n}(d^{n}-1)=\frac{1}{2}(d^{2n}-d^{2n-1}-d^{2n-2}-d^{n+1}+d^{n}+d^{n-1}-d^2+d)$ is greater than zero for any $n \ge 2$ and $d \ge 2$, which implies $\gamma^{(ij)}=0$ in the even $N$-qudit case.
			Thus $\ket{\phi}=e^{\sqrt{-1}\varphi}\ket{\psi}$ which corresponds to the same physical state.
		\end{proof}
		
		
Note that the proof of Theorem~\ref{theorem111} is valid
for generic states only. In the following two sections, we investigate UDP problems for special states.
		
		\section{Lower bounds on number of marginals for UDP}\label{sec:hypergraphs}
		In the previous sections, we showed that almost all even $N$-qudit pure states $\ket{\psi}$ can be UDP by its $\mathcal{D}_{N/2}(\ket{\psi})$. Specifically, only four marginals in $\mathcal{D}_{N/2}(\ket{\psi})$ are required. In this section, we will discuss the low bound of the number of marginals that are necessary to uniquely determine the state $\ket{\psi}$. The \emph{marginal number} for $\ket{\psi}$ to be UDP, denoted by M($\mathcal{D}_{k}(\ket{\psi})$), is the smallest number of $k$-body marginals that can uniquely determine the pure state $\ket{\psi}$ among pure state.

		Now we describe a deck by a hypergraph. Given a deck by ${\mathcal{F}}$, we define a hypergraph $G_{\mathcal{F}}$ as follows: the vertex set is $\mathcal{I}$ and the edge set is $\mathcal{F}$. That is, each particle corresponds to a vertex, and each marginal corresponds to an edge; a particle belongs to a marginal means that the corresponding vertex  belongs to the corresponding edge.
    	A hypergraph is said to be connected if there is a path between every pair of vertices; here a path means a sequence of edges $e_1,\ldots,e_s$ and a sequence of distinct vertices $v_1,\ldots,v_{s+1}$ such that $\{v_1,v_{i+1}\}\subset e_i$ for all $i=1,\ldots,s$. Otherwise, the graph is disconnected.
	    For example, if $\cF=\{\{123\}, \{456\}, \{124\}, \{356 \}\}$ is given in Fig.~\ref{fig2}, then it is connected. Further, if each element in $\mathcal{F}$ has size $k$, then we denote $G_{\mathcal{F}}$ as $k$-$G_{\mathcal{F}}$.
	
	    Next, we show that the UDP of certain pure states has relations with the connectivity of the graph.
	
	
	We call an $N$-particle pure state $\ket{\psi}$ fully product if it can be written as $\ket{\psi}= \ket{\alpha}_{1}\otimes\ket{\beta}_{2}\otimes\cdots\otimes\ket{\gamma}_{{N}}\in \mathcal{H}_{1}\otimes \cdots\otimes \mathcal{H}_{N}\cong (\mathbb{C}^d)^{\otimes N}$.
	Otherwise, it is said to be entangled. A multipartite
	state $\ket{\psi}$ is said to be \emph{genuinely multipartite entangled} (GME), if it cannot be written in a biproduct form, i.e., $\ket{\psi}\neq\ket{\alpha}_{\mathcal J}\otimes\ket{\beta}_{\mathcal J_C}$ for any bipartite cut ${\mathcal J}|{\mathcal J_C}$. We denote $\dim \mathcal J=\dim\otimes_{i\in \cJ} \cH_{i}$.
	
	\begin{theorem}\label{connected}
		For any $N$-particle GME pure state, if the state is $\mathcal{F}$-UDP, then $G_{\mathcal{F}}$ must be connected.
		
	\end{theorem}
	\begin{proof}
		We prove it by contradiction. Let $\ket{\psi}$ be an $N$-particle GME pure state, which is $\mathcal{F}$-UDP. Suppose that $G_{\mathcal{F}}$ is disconnected.  We claim that there exists a partition of its vertex set into two nonempty sets $X$ and $Y$ such that there is no edge across $X$ and $Y$. Otherwise, for any two vertices $x$ and $y$, let $X=\{x\}$, we can always find a vertex $x_1$ such that $\{x,x_1\}$ belongs to an edge. Then we let $X_1=\{x, x_1\}$, and by assumption there exists a vertex $x_2$ such that either $\{x,x_2\}$ or $\{x_1,x_2\}$ belongs to an edge. This implies that there is a path between $x$ and $x_2$. Continuing in this way, we can get a path from $x$ and $y$.
		
		Therefore,
		there exists a bipartition ${\mathcal J}|{\mathcal J_C}$ such that each subset in $\mathcal{F}$ is contained either in $\mathcal J$ or $\mathcal J_C$. Let $d=\min(\dim \mathcal J,\dim \mathcal J_C)$. Using the Schmidt decomposition along this bipartition, we can write
		\begin{equation}
			\ket{\psi}= \sum_{i=1}^{d}\sqrt{\lambda_i}\ket{i}_{\mathcal J}\ket{i}_{\mathcal J_C}.
		\end{equation}
		Consequently,
		\begin{equation}
			\rho_{\mathcal J} = \text{Tr}_{\mathcal J_C}\ket{\psi}\bra{\psi}
			=\sum_{i=1}^{d}{\lambda_i}\ket{i}\bra{i}_{\mathcal J},
		\end{equation}
		and
		\begin{equation}
			\rho_{\mathcal J_C} = \text{Tr}_{\mathcal J}\ket{\psi}\bra{\psi}
			=\sum_{i=1}^{d}{\lambda_i}\ket{i}\bra{i}_{\mathcal J_C}. \\
		\end{equation}
		Since $\ket{\psi}$ is a GME pure state, $d$ is at least two. 
		Therefore, we can change the phase of different $i$ in order to get another state $\ket{\phi}$ which has the same deck by $\mathcal{F}$ as $\ket{\psi}$,
		\begin{equation}
			\ket{\phi}= \sum_{i=1}^{d}e^{\sqrt{-1}\varphi_i}\sqrt{\lambda_i}\ket{i}_{\mathcal J}\ket{i}_{\mathcal J_C}.
		\end{equation}
	\end{proof}
	
	Theorem \ref{connected} gives a necessary condition of the deck for a GME pure state to be UDP, which can give a lower bound on the number of marginals for which a GME state can be $k$-UDP. For example, Parashar and Rana in \cite{PhysRevA.80.012319} proved that the generalized $N$-qubit $W$ state $\ket{W}_N$ is UDA (thus UDP) by just ($N-1$) marginals of $\mathcal{D}_{2}(\ket{W}_N)$. We know that the hypergraph corresponding to $\mathcal{D}_{2}(\rho)$ is a simple graph, and a simple graph means that each edge is determined by exactly two vertices. By a simple conclusion in graph theory, we know that the connected graph of $N$ vertices with a minimum number of edges is a tree, which has $(N-1)$ edges. This answers several questions in \cite{PhysRevA.80.012319}.
	 One of the questions posed in \cite{PhysRevA.80.012319} is to ask whether $\ket{W}_N$ can uniquely be determined by arbitrary ($N-1$) 2-body marginals. Wu $\emph{et al.}$ showed that $\ket{W}_N$ can be $\mathcal{F}$-UDP by their ($N-1$) 2-body marginals, if the $G_{\mathcal{F}}$ is a tree \cite{PhysRevA.90.012317}. Since $\ket{W}_N$ is a GME pure state, Theorem \ref{connected} answers this question in a negative way: the graph from the $(N-1)$ 2-body marginals must be connected. This means that if $\ket{W}_N$ can be $\mathcal{F}$-UDP by ($N-1$) 2-body marginals, the $G_{\mathcal{F}}$ must be a tree. Another question  in \cite{PhysRevA.80.012319} is to ask if  $\ket{W}_N$ can be determined by less number of
	2-body marginals. Theorem \ref{connected} also gives a negative answer, ($N-1$) is the optimal number since a connected graph on $N$ vertices has at least ($N-1$) edges. For the general case, we get the following lower bounds.

			\begin{corollary}\label{cor222}
	For any $N$-particle GME pure state $\ket{\phi}$, at least $\lceil \frac{N-1}{k-1}  \rceil$ $k$-body marginals are required for $\ket{\phi}$ to be $k$-UDP. This implies that $M(\mathcal{D}_{k}(\ket{\phi})) \ge \lceil \frac{N-1}{k-1}  \rceil$.
	
	\end{corollary}
\begin{proof}
	Suppose there exist $(\lceil \frac{N-1}{k-1} \rceil-1)$ $k$-body marginals can uniquely determine an $N$-particle GME pure state. Then we have a $k$-deck by $\mathcal{F}$ with $|\mathcal{F}|=(\lceil \frac{N-1}{k-1}  \rceil-1)$ such that $G_{\mathcal{F}}$ is connected.
    However, we claim that this is impossible.  If the graph is connected, there exists an ordering of all edges $e_1,\ldots,e_{|\mathcal{F}|}$ such that each edge $e_j$ intersects with at least one edge $e_i$ with $i<j$. This implies that the union of all edges can not cover all vertices, since
%
	\begin{equation}
		k+(\lceil \frac{N-1}{k-1}  \rceil-2)\times (k-1) = \lceil \frac{N-1}{k-1} \rceil \times (k-1) -(k-2) \le N-1
	\end{equation}
	That's a contradiction.

\end{proof}

		Rana $\emph{et al.}$ proved that the $N$-qubit generic Dicke states $\ket{GD(N,l)}$ are UDP by their $(l+1)$-body marginals, and $\binom{N-1}{l}$ many of them are sufficient \cite{PhysRevA.84.052331}. This implies that M($\mathcal{D}_{l+1}(\ket{GD(N,l)})$) is no more than $\binom{N-1}{l}$. Wu $\emph{et al.}$ showed
		that no more than $\binom{N}{l}/2$ ($l+1$)-body marginals can uniquely determine the $\ket{GD(N,l)}$ \cite{PhysRevA.92.052338}. This means that M($\mathcal{D}_{l+1}(\ket{GD(N,l)})$) is no more than $\binom{N}{l}/2$. Note that  $\binom{N}{l}/2\leq \binom{N-1}{l}$ and the equality holds if and only if $l=N/2$. In other words, these results indicate that the upper bound of M($\mathcal{D}_{l+1}(\ket{GD(N,l)})$) will not exceed $\binom{N}{l}/2$. However, the lower bound obtained by Corollary \ref{cor222} is M($\mathcal{D}_{l+1}(\ket{GD(N,l)})$) $\ge \lceil \frac{N-1}{l}  \rceil$, which is a little bit far from the upper bound.

			\section{Exploration of states that are not UDP}\label{sec:counterexample}
It is  known that almost all $N$ qudit pure states can be UDA by its $\lceil N/2 \rceil$-body marginals \cite{PhysRevA.71.012324,PhysRevLett.89.207901,huang2018quantum}, and UDP by its $\lfloor N/2 \rfloor$-body marginals from Theorem~\ref{theorem111}. In this section, we give examples of pure states that cannot be UDP, and hence cannot be UDA, by their $k$-body marginals for some $k\geq N/2$.


			An  $N$-particle pure state in $(\bbC^{d})^{\otimes N}$
			is called a $k$-uniform state, if all its reductions to $k$ parties are maximally mixed. By Schmidt decomposition,  $k$-uniform states only exist when $k\leq \lfloor N/2\rfloor$. There are many constructions of  $k$-uniform states  \cite{PhysRevA.87.012319},
			so $k$-uniform states  are clearly not UDP by their $\mathcal{D}_{k}(\rho)$. In the following, we find that some $k$-uniform states are not even UDP by their $\mathcal{D}_{N-k}(\rho)$, which has higher-order since $k\leq N/2$.
		 More interestingly,  by modifying these $k$-uniform quantum states, we obtain states that are not $k$-uniform anymore, but still preserve the properties that they cannot be UDP by their $\mathcal{D}_{N-k}(\rho)$.
			
			For example, we  know that the GHZ state, as a $1$-uniform state, cannot be UDP by its $\mathcal{D}_{N-1}(\rho)$ \cite{PhysRevLett.100.050501}. The generalized GHZ state
			\begin{equation}\label{eqghz}
				\alpha\ket{00\cdots0}+\beta\ket{11\cdots1},\text{ where }\alpha,\beta \neq 0,
			\end{equation}
			 is not 1-unform any more when $\alpha \neq \beta$, but still cannot be UDP by its full quantum marginals \cite{PhysRevLett.100.050501}.

            Here is another example. The following is a $2$-uniform states in $(\bbC^3)^{\otimes 4}$,
            \begin{equation}\label{phi}
	              \begin{aligned}
		               \ket{\phi}=&\frac{1}{3}(\ket{0000}+\ket{0111}+\ket{0222}+\ket{1021}+\ket{1102}\\&+\ket{1210}+\ket{2012}+\ket{2120}+\ket{2201}),
	              \end{aligned}
            \end{equation}
            which is clearly not UDP by its $2$-deck.
            We define its generalized state of the form
			\begin{equation}\label{genphi}
				\begin{aligned} \ket{\phi'}=&a_1\ket{0000}+a_2\ket{0111}+a_3\ket{0222}+a_4\ket{1021}\\&+a_5\ket{1102}+a_6\ket{1210}+a_7\ket{2012}\\&+a_8\ket{2120}+a_9\ket{2201},
				\end{aligned}				
			\end{equation}
	         where $a_i\neq0$, $i=1,\ldots,9$.  Note that to any two parties, the reduction of $\ket{\phi'}$ is of the form $\sum_{i_1,i_2}b_{i_1i_2}\ket{i_1i_2}\bra{i_1i_2}$, where $b_{i_1i_2}=a^2_i$ for the unique coefficient $a_i$ whose corresponding basis state reduces to $\ket{i_1i_2}$ on the two parties. So the phase change on any basis state gives another state with the same $2$-deck. For example, the state $\ket{\phi''}$ obtained from $\ket{\phi'}$ by changing $a_1$ to $-a_1$ is another state with the same $2$-deck as $\ket{\phi'}$.
%
%
			Note that this example also illustrates that there exist generic pure states that cannot be determined by half-body marginals.

Next, we extend the above examples into a family of states with $N$-particles that are not UDP by their $(N-k)$-deck with $k\leq N/2$.
			
			In \cite{PhysRevA.97.062326}, Goyeneche $\emph{et al.}$ generalized some classical combinatorial designs into quantum versions, including a generalization of the classical \emph{orthogonal array} (OA) to \emph{quantum orthogonal array} (QOA). In the following we will use QOA to give a class of $N$-qudit states obtained from OA with strength $k$ that cannot be UDP by their $\mathcal{D}_{N-k}$.
			
			An \emph{orthogonal array} with level $d$, strength $k$ and index $\lambda$, denoted as OA$(r, N, d, k)$ with $r=\lambda d^k$, is an $r \times N$ array $A$ with entries from the set $S=\{0, \ldots, d-1 \}$  such that  every $r \times k$ subarray of $A$ contains each $k$-tuple over $S$ exactly $\lambda$ times as a row \cite{Hedayat1999}.
			An OA is called irredundant (IrOA) if every  $r\times (N-k)$
			subarray contains no repeated rows \cite{PhysRevA.90.022316}. Clearly, any OA with index one is irredundant.
			
\begin{example}\label{egoa}The following is  an OA$(9,4,3,2)$ with index $1$,
			\begin{equation}
				\left[\begin{array}{ll}
					0000  \\
					0111  \\
					0222  \\
					1021\\1102\\1210\\2012\\2120\\2201\\
				\end{array}\right].
			\end{equation}
\end{example}			

			
			\begin{definition}\cite{PhysRevA.97.062326}
				A quantum orthogonal array QOA$(r, N, d, k)$ is an arrangement consisting of $r$ rows composed by $N$-particle pure quantum states
				$\ket{\psi_i}\in\mathcal{H}_{1}\otimes \cdots\otimes \mathcal{H}_{N}\cong (\mathbb{C}^d)^{\otimes N}$ if
				\begin{equation}
\sum_{i,j\in[r]}Tr_{A}(\ket{\psi_i}\bra{\psi_j})=\frac{r}{d^k}I_{d^k}
				\end{equation}
				for every subset $A$ of $N - k$ parties.
			\end{definition}
			
			A QOA$(r, N, d, k)$ can be obtained from an IrOA$(r, N, d, k)$ by writing each row $(i_1, i_2, \ldots, i_N)$ as a quantum item $\ket{i_1 i_2 \ldots i_N}$. So a QOA$(9,4,3,2)$ from Example \ref{egoa} is as follows,
\begin{equation}
				\left[\begin{array}{ll}
					\ket{0000}  \\
					\ket{0111}  \\
					\ket{0222}  \\
					\ket{1021}\\
					\ket{1102}\\
					\ket{1210}\\
					\ket{2012}\\
					\ket{2120}\\\ket{2201}\\
				\end{array}\right].
			\end{equation}
By definition, a QOA$(r, N, d, k)$ can give rise to an $N$-particle $k$-uniform state $\ket{\phi}=\frac{1}{\sqrt{r}}\sum_{i\in [r]}\ket{\psi_i}$, with local dimension $d$ and $r$ items \cite{PhysRevA.97.062326}. Note that the $2$-uniform state $\ket{\phi}$ mentioned in Eq.~\eqref{phi} is  obtained
			from the above QOA$(9,4,3,2)$ directly.

			 Similar to the generalized GHZ state in Eq. (\ref{eqghz}), and the generalized state in Eq. (\ref{genphi}), for each quantum state $\ket{\phi}=\frac{1}{\sqrt{r}}\sum_{i\in [r]}\ket{\psi_i}$ corresponding to a QOA$(r, N, d, k)$,  we define its generalized form \[\ket{\phi'}=\sum_{i\in [r]}a_i\ket{\psi_i}\text{ with } a_i \neq 0.\] We call such $\ket{\phi'}$ as the generalized QOA$(r, N, d, k)$ state. 	
		\begin{theorem}\label{QOA}
			The generalized QOA$(r, N, d, k)$ state obtained from an OA$(r, N, d, k)$ with index $1$ cannot be UDP by its $(N-k)$-deck for $k\le \lfloor N/2 \rfloor$.
		\end{theorem}
				\begin{proof} Since the generalized QOA$(r, N, d, k)$ state is obtained from an OA$(r, N, d, k)$ with index $1$, we can write it as
			$\ket{\psi}=\sum_{i=1}^r a_i\ket{i_1i_2 \cdots i_N}$. By the definition of OA, for any $N-k$ parties, the $(N-k)$-body marginal can be expressed as
			\begin{equation}\label{nkbody}
				\sum_{i=1}^r b_{i}\ket{i_1i_2 \cdots i_{N-k}}\bra{i_1i_2 \cdots i_{N-k}},
			\end{equation}
			where $b_{i}$ is  some $a_{i}^2$
			for the unique coefficient $a_i$ whose corresponding basis state reduces to $\ket{i_1i_2 \cdots i_{N-k}}$ on the $N-k$ parties.
			It is easy to check that $\ket{\psi'}=\sum_{i=1}^r e^{\sqrt{-1}\varphi_{i}}{a_{i}}\ket{i_1i_2 \cdots i_N}$ have the same complete $(N-k)$-deck as $\ket{\psi}$, but $\ket{\psi'}$ is a quantum state different from $\ket{\psi}$ if the phases are different.			
		\end{proof}
		
By Theorem~\ref{QOA}, the QOA state $\ket{\phi}=\frac{1}{\sqrt{r}}\sum_{i\in [r]}\ket{\psi_i}$ obtained from an OA$(r, N, d, k)$ with index $1$ is a $k$-uniform state which cannot be UDP by its $(N-k)$-deck, while the generalised QOA state $\ket{\psi}=\sum_{i=1}^r a_i\ket{i_1i_2 \cdots i_N}$ is not $k$-uniform in general, but still cannot be UDP by its $(N-k)$-deck.

In the proof of Theorem~\ref{QOA}, the key point that grantees us to find a different state with the same $(N-k)$-deck is  Eq. (\ref{nkbody}): each $(N-k)$-body reduction has this form with $b_i$ being some $a_i^2$. However, to get each reduction like Eq. (\ref{nkbody}), we don't need each $k$-tuple appearing \emph{exactly} once as required by an OA with index $1$;  we only need each $k$-tuple appearing \emph{at most} once. We call such an array by \emph{packing array}, denoted by  PA$(r, N, d, k)$. Here $2\leq r\leq d^k$, and  $r=d^k$ if it is an OA with index $1$. Then Theorem~\ref{QOA} can be extended as follows with exactly the same proof.

			\begin{theorem}\label{QPA}
				The state $\ket{\psi}=\sum_{i=1}^r a_i\ket{i_1i_2 \cdots i_N}$ with $a_i \neq 0$  obtained from a PA$(r, N, d, k)$ cannot be UDP by its $(N-k)$-deck, where $k\leq \lfloor N/2 \rfloor$.
			\end{theorem}

Note that when $k\leq \lfloor N/2 \rfloor$, Theorem~\ref{QPA} provides a large amount of $N$-qudit states that are not UDP by its $(N-k)$-deck, since  a PA$(r, N, d, k)$ has a very weak structure and commenly exists.

		\section{Conclusion}\label{sec:con}
		In summary, we show that almost all even-qudit pure states can be UDP by its half-body deck, and give a class of states obtained from OA that cannot be UDP.
		Meanwhile, we establish a necessary condition for GME states to be UDP. Interestingly, this leads to some lower bounds of M($\mathcal{D}_{k}(\ket{\psi})$), i.e., the minimum number of $k$-body marginals required to determine the state.
		Based on the results we have obtained, we can see that whether a quantum state can be uniquely determined by its $k$-deck is strongly related to the nature of the quantum state itself. This can also be seen from the long study of the uniqueness problem in QMP. It is known that only GHZ states and their LU equivalence class cannot be UDP by its $k$-deck with $k\leq N-1$ among arbitrary $N$-qudit states, while $W$ states are even UDA  by its $2$-deck. The problem of whether a quantum state can be UDP by its  $k$-deck for $k<\lfloor N/2 \rfloor$ still opens. In addition, an explicit classification of whether quantum states can be uniquely determined by their $k$-deck remains a matter of interest. There are also works on the minimum number of elements in a $k$-deck  by which pure states are UDP or UDA \cite{PhysRevA.80.012319,PhysRevA.92.052338}. For this problem, we give a necessary condition that can be UDP for GME pure states and show some existing results whose decks do reach a minimum size. But for most  quantum states, the lower bound of M($\mathcal{D}_{k}(\ket{\psi})$) is still unknown.
		
		\section{Acknowledgments}
		The authors thank Nikolai Wyderka and Felix Huber for some helpful discussions and suggestions.
		This work is supported by the NSFC
		under Grants No. 12171452 and No. 12231014, the Innovation Program for Quantum Science and Technology
		(2021ZD0302902) and the National Key Research and
		Development Program of China (2020YFA0713100).
		
		\appendix

\bibliographystyle{IEEEtran}
\bibliography{marginal}

\begin{thebibliography}{10}
\providecommand{\url}[1]{#1}
\csname url@samestyle\endcsname
\providecommand{\newblock}{\relax}
\providecommand{\bibinfo}[2]{#2}
\providecommand{\BIBentrySTDinterwordspacing}{\spaceskip=0pt\relax}
\providecommand{\BIBentryALTinterwordstretchfactor}{4}
\providecommand{\BIBentryALTinterwordspacing}{\spaceskip=\fontdimen2\font plus
\BIBentryALTinterwordstretchfactor\fontdimen3\font minus
  \fontdimen4\font\relax}
\providecommand{\BIBforeignlanguage}[2]{{%
\expandafter\ifx\csname l@#1\endcsname\relax
\typeout{** WARNING: IEEEtran.bst: No hyphenation pattern has been}%
\typeout{** loaded for the language `#1'. Using the pattern for}%
\typeout{** the default language instead.}%
\else
\language=\csname l@#1\endcsname
\fi
#2}}
\providecommand{\BIBdecl}{\relax}
\BIBdecl

\bibitem{Klyachko_2006}
A.~A. Klyachko, ``Quantum marginal problem and {$N$}-representability,''
  \emph{Journal of Physics: Conference Series}, vol.~36, no.~1, p.~72, 2006.

\bibitem{schilling2015quantum}
C.~Schilling, ``The quantum marginal problem,'' \emph{Mathematical Results in
  Quantum Mechanics: Proceedings of the QMath12 Conference}, pp. 165--176,
  2015.

\bibitem{haapasalo2021quantum}
E.~Haapasalo, T.~Kraft, N.~Miklin, and R.~Uola, ``Quantum marginal problem and
  incompatibility,'' \emph{Quantum}, vol.~5, p. 476, 2021.

\bibitem{Schilling2017Quantum}
C.~Schilling, ``Quantum marginal problem and its physical relevance,''
  \emph{arXiv:1507.00299}.

\bibitem{QMP}
T.~Tyc and J.~Vlach, ``Quantum marginal problems,'' \emph{The European Physical
  Journal D}, vol.~69, 2015.

\bibitem{RevModPhys.35.668}
A.~J. Coleman, ``Structure of {F}ermion {D}ensity {M}atrices,'' \emph{Reviews
  of modern Physics}, vol.~35, pp. 668--686, 1963.

\bibitem{PhysRevLett.118.020401}
T.~Xin, D.~Lu, J.~Klassen, N.~Yu, Z.~Ji, J.~Chen, X.~Ma, G.~Long, B.~Zeng, and
  R.~Laflamme, ``Quantum {S}tate {T}omography via {R}educed {D}ensity
  {M}atrices,'' \emph{Phys. Rev. Lett.}, vol. 118, p. 020401, 2017.

\bibitem{PhysRevLett.124.100401}
J.~Cotler and F.~Wilczek, ``Quantum overlapping tomography,'' \emph{Phys. Rev.
  Lett.}, vol. 124, p. 100401, 2020.

\bibitem{PhysRevLett.89.207901}
N.~Linden, S.~Popescu, and W.~K. Wootters, ``Almost every pure state of three
  qubits is completely determined by its two-particle reduced density
  matrices,'' \emph{Phys. Rev. Lett.}, vol.~89, p. 207901, 2002.

\bibitem{PhysRevLett.100.050501}
S.~N. Walck and D.~W. Lyons, ``Only {$n$}-qubit
  {G}reenberger-{H}orne-{Z}eilinger states are undetermined by their reduced
  density matrices,'' \emph{Phys. Rev. Lett.}, vol. 100, p. 050501, 2008.

\bibitem{PhysRevA.79.032326}
S.~N. {}Walck and D.~W. Lyons, ``Only {$n$}-qubit
  {G}reenberger-{H}orne-{Z}eilinger states contain {$n$}-partite information,''
  \emph{Phys. Rev. A}, vol.~79, p. 032326, 2009.

\bibitem{PhysRevA.71.012324}
N.~S. Jones and N.~Linden, ``Parts of quantum states,'' \emph{Phys. Rev. A},
  vol.~71, p. 012324, 2005.

\bibitem{huang2018quantum}
S.~Huang, J.~Chen, Y.~Li, and B.~Zeng, ``Quantum state tomography for generic
  pure states,'' \emph{SCIENCE CHINA Physics, Mechanics $\&$ Astronomy},
  vol.~61, no.~11, pp. 1--7, 2018.

\bibitem{Huber_2018}
F.~Huber and S.~Severini, ``Some {U}lam’s reconstruction problems for quantum
  states,'' \emph{Journal of Physics A: Mathematical and Theoretical}, vol.~51,
  no.~43, p. 435301, 2018.

\bibitem{Almostgraphthree}
B.~Bollobás, ``Almost every graph has reconstruction number three,''
  \emph{Journal of Graph Theory}, vol.~14, no.~1, pp. 1--4, 1990.

\bibitem{PhysRevA.70.010302}
L.~Di\'osi, ``Three-party pure quantum states are determined by two two-party
  reduced states,'' \emph{Phys. Rev. A}, vol.~70, p. 010302, 2004.

\bibitem{2012Comment}
J.~Chen, Z.~Ji, M.~B. Ruskai, B.~Zeng, and D.~Zhou, ``Comment on some results
  of erdahl and the convex structure of reduced density matrices,''
  \emph{Journal of Mathematical Physics}, vol.~53, no.~7, pp. 1608--1621, 2012.

\bibitem{PhysRevA.96.010102}
N.~Wyderka, F.~Huber, and O.~G\"uhne, ``Almost all four-particle pure states
  are determined by their two-body marginals,'' \emph{Phys. Rev. A}, vol.~96,
  p. 010102, 2017.

\bibitem{LLOYD1988186}
S.~Lloyd and H.~Pagels, ``Complexity as thermodynamic depth,'' \emph{Annals of
  physics}, vol. 188, no.~1, pp. 186--213, 1988.

\bibitem{AJScott_2003}
A.~J. Scott and C.~M. Caves, ``Entangling power of the quantum baker's map,''
  \emph{Journal of Physics A: Mathematical and General}, vol.~36, no.~36, p.
  9553, aug 2003.

\bibitem{PhysRevA.80.012319}
P.~Parashar and S.~Rana, ``{$N$}-qubit {$W$} states are determined by their
  bipartite marginals,'' \emph{Phys. Rev. A}, vol.~80, p. 012319, 2009.

\bibitem{PhysRevA.90.012317}
X.~Wu, G.~Tian, W.~Huang, Q.~Wen, S.~Qin, and F.~Gao, ``Determination of {$W$}
  states equivalent under stochastic local operations and classical
  communication by their bipartite reduced density matrices with tree form,''
  \emph{Phys. Rev. A}, vol.~90, p. 012317, 2014.

\bibitem{PhysRevA.84.052331}
S.~Rana and P.~Parashar, ``Optimal reducibility of all $w$ states equivalent
  under stochastic local operations and classical communication,'' \emph{Phys.
  Rev. A}, vol.~84, p. 052331, 2011.

\bibitem{PhysRevA.92.052338}
X.~Wu, Y.~Yang, Q.~Wen, S.~Qin, and F.~Gao, ``Determination of {D}icke states
  equivalent under stochastic local operations and classical communication,''
  \emph{Phys. Rev. A}, vol.~92, p. 052338, 2015.

\bibitem{PhysRevA.87.012319}
L.~Arnaud and N.~J. Cerf, ``Exploring pure quantum states with maximally mixed
  reductions,'' \emph{Phys. Rev. A}, vol.~87, p. 012319, 2013.

\bibitem{PhysRevA.97.062326}
D.~Goyeneche, Z.~Raissi, S.~Di~Martino, and K.~\ifmmode~\dot{Z}\else
  \.{Z}\fi{}yczkowski, ``Entanglement and quantum combinatorial designs,''
  \emph{Phys. Rev. A}, vol.~97, p. 062326, 2018.

\bibitem{Hedayat1999}
A.~S. Hedayat, N.~J.~A. Sloane, and J.~Stufken, \emph{Orthogonal Arrays: Theory
  and Applications}.\hskip 1em plus 0.5em minus 0.4em\relax Springer New York,
  1999.

\bibitem{PhysRevA.90.022316}
D.~Goyeneche and K.~\ifmmode~\dot{Z}\else \.{Z}\fi{}yczkowski, ``Genuinely
  multipartite entangled states and orthogonal arrays,'' \emph{Phys. Rev. A},
  vol.~90, p. 022316, 2014.

\end{thebibliography}

\end{document}